\documentclass[a4paper,UKenglish,cleveref, autoref, thm-restate]{lipics-v2021}

\nolinenumbers

\usepackage[linesnumbered,ruled,vlined]{algorithm2e}
\usepackage{etoolbox}

\SetKw{Continue}{continue}

\newcommand{\head}[1]{\textsc{head}\ifstrempty{#1}{}{\ifstrequal{#1}{(}{(}{(#1)}}}
\newcommand{\tail}[1]{\textsc{tail}\ifstrempty{#1}{}{\ifstrequal{#1}{(}{(}{(#1)}}}
\newcommand{\wt}[1]{\textsc{start}\ifstrempty{#1}{}{\ifstrequal{#1}{(}{(}{(#1)}}}
\newcommand{\wh}[1]{\textsc{end}\ifstrempty{#1}{}{\ifstrequal{#1}{(}{(}{(#1)}}}

\title{Cut paths and their remainder structure, with applications} %

\author{Massimo Cairo}{Department of Computer Science, University of Helsinki, Finland}{}{}{}%

\author{Shahbaz Khan\footnote{Most of the work by the author was done while he was affiliated to University of Helsinki.}}{
Department of Computer Science and Engineering, Indian Institute of Technology Roorkee, India}{shahbaz.khan@cs.iitr.ac.in}
{https://orcid.org/0000-0001-9352-0088}{}

\author{Romeo Rizzi}{Department of Computer Science, University of Verona, Italy}{romeo.rizzi@univr.it}{https://orcid.org/0000-0002-2387-0952}{}

\author{Sebastian Schmidt}{Department of Informatics, Department of Computer Science, University of Helsinki, Finland}{sebastian.schmidt@helsinki.fi}{https://orcid.org/0000-0003-4878-2809}{}

\author{Alexandru~I.~Tomescu}{Department of Computer Science, University of Helsinki, Finland}{alexandru.tomescu@helsinki.fi}{https://orcid.org/0000-0002-5747-8350}{}

\author{Elia~C.~Zirondelli}{Department of Mathematics, University of Trento, Italy}{eliacarlo.zirondelli@unitn.it}{}{}

\authorrunning{M. Cairo, S. Khan, R. Rizzi, S. Schmidt, A. I. Tomescu and E. C. Zirondelli} %

\Copyright{Massimo Cairo, Shahbaz Khan, Romeo Rizzi, Sebastian Schmidt, Alexandru. I. Tomescu and Elia C. Zirondelli} %



\funding{\textit{Alexandru~I.~Tomescu}: This work was partially funded by the Academy of Finland (grants No.~322595, 328877). \textit{Shahbaz Khan, Sebastian Schmidt and Alexandru~I.~Tomescu}: This work was partially funded by the European Research Council (ERC) under the European Union's Horizon 2020 research and innovation programme (grant agreement No.~851093, SAFEBIO).}%

\EventEditors{John Q. Open and Joan R. Access}
\EventNoEds{2}
\EventLongTitle{42nd Conference on Very Important Topics (CVIT 2016)}
\EventShortTitle{CVIT 2016}
\EventAcronym{CVIT}
\EventYear{2016}
\EventDate{December 24--27, 2016}
\EventLocation{Little Whinging, United Kingdom}
\EventLogo{}
\SeriesVolume{42}
\ArticleNo{23}

\begin{document}

\maketitle

\begin{abstract}
In a strongly connected graph $G = (V,E)$, a \emph{cut arc} (also called \emph{strong bridge}) is an arc $e \in E$ whose removal makes the graph no longer strongly connected. Equivalently, there exist $u,v \in V$, such that all $u$-$v$ walks contain $e$. Cut arcs are a fundamental graph-theoretic notion, with countless applications, especially in reachability problems. 

In this paper we initiate the study of \emph{cut paths}, as a generalisation of cut arcs, which we naturally define as those paths $P$ for which there exist $u,v \in V$, such that all $u$-$v$ walks contain $P$ as subwalk. We first prove various properties of cut paths and define their remainder structures, which we use to present a simple $O(m)$-time verification algorithm for a cut path ($|V| = n$, $|E| = m$). 

Secondly, we apply cut paths and their remainder structures to improve several reachability problems from bioinformatics, as follows. A walk is called \emph{safe} if it is a subwalk of every node-covering closed walk of a strongly connected graph. \emph{Multi-safety} is defined analogously, by considering node-covering \emph{sets} of closed walks instead. We show that cut paths provide \emph{simple} $O(m)$-time algorithms verifying if a walk is safe or multi-safe. For multi-safety, we present the first linear time algorithm, while for safety, we present a simple algorithm where the state-of-the-art employed complex data structures.
Finally we show that the simultaneous computation of remainder structures of all subwalks of a cut path can be performed in linear time, since they are related in a structured way.
These properties yield an $O(mn)$-time algorithm outputting all maximal multi-safe walks, improving over the state-of-the-art algorithm running in time $O(m^2+n^3)$.

The results of this paper only scratch the surface in the study of cut paths, and we believe a rich structure of a graph can be revealed, considering the perspective of a path, instead of just an arc.
\end{abstract}

\section{Introduction}

\subsection{Motivation}

Connectivity problems are a fundamental aspect of graph theory and graph algorithms.
In directed graphs, \emph{cut arcs} (also known as \emph{strong bridges}) are a basic structure to characterise the reachability properties of the graph.
They are defined as arcs whose removal makes the graph no longer strongly connected, or equivalently, as arcs $e$ such that there exists a pair of nodes $u,v$ such that each $u\text{-}v$ walk contains $e$.
Cut arcs and the related strongly connected components are nowadays part of any lecture about graph theory.
Moreover, significant work has been done to investigate the properties of a graph in relation to its cut arcs, e.g.~by finding all cut arcs in linear time~\cite{italiano2012finding}, and, after linear-time preprocessing, answering connectivity queries in constant time under the removal of any single arc~\cite{georgiadis2020strong}.
These gave rise to further theoretical advances, and are used in algorithms to compute e.g. 2-vertex connected components in directed graphs~\cite{georgiadis20182vertex} or 2-edge connected components in directed graphs~\cite{georgiadis20162edge}.
The results are also useful for more practical works, like in the analysis of non-equilibrium biochemical reaction networks~\cite{yordanov2020efficient} or when analysing real-world graphs such as social networks or the world wide web~\cite{italiano2017bowtie}.
Naturally, cut arcs play a crucial role in various practical networks, as they represent critical links in e.g.~communication networks, road networks or transportation networks.

A natural generalisation of a cut arc is a \emph{cut path}\footnote{Note that walks that contain a cycle cannot be cut paths, hence the name \emph{cut path} and not \emph{cut walk}.}, similarly defined as a walk (not a single arc) $W$ such that there exists a pair of nodes $u,v$ such that each $u\text{-}v$ walk has $W$ as subwalk. 
While there is (to the best of our knowledge) no research around cut paths as such, they seem to be equally fundamental as cut arcs. On the practical side, one can view cut paths 
as representing critical \emph{routes} through social networks, the world wide web, communication networks, road networks or transportation networks.
Additionally, in practical applications, when nodes represent street crossings or network routers, cut paths imply critical \emph{links} within these objects.
On the more theoretical side, in this paper we show that cut paths are a useful tool in some reachability problems theoretically modelling the genome assembly problem in bioinformatics. 
In addition to exhibiting interesting properties on their own, cut paths also allows us to improve several of these results, as we discuss in \Cref{sec:intro-applications}.

\subsection{Overview of cut paths}
To give an overview of cut paths, we start with some basic definitions.
A \emph{graph} $G = (V, E)$ with $n$ nodes and $m$ arcs is directed and may have self-loops. For an arc $e = (u, v)$, we call $u = \tail(e)$ its \emph{tail} and $v = \head(e)$ its \emph{head}.
Sets of nodes can induce subgraphs in the standard manner.
A graph is \emph{strongly connected} if each pair of nodes is connected by a directed path in both directions.
A \emph{strongly connected component (SCC)} is a maximal induced subgraph that is strongly connected.
A \emph{cut arc} is an arc that upon removal increases the number of strongly connected components in $G$. 

An equivalent definition of a cut arc in strongly connected graphs is an arc that is part of all walks from some node to some other node.
Generalising, a \emph{cut path} is a walk that is a subwalk of all walks from some node to some other node.
We assume a graph to be strongly connected from here on.

When removing a cut arc $(u,v)$ from a graph, the graph is not strongly connected anymore, but instead contains multiple SCCs.
There is exactly one source SCC (containing $v$), which is an SCC without any incoming arcs from other SCCs, and exactly one sink SCC (containing $u$), which is an SCC without any outgoing arcs to other SCCs.
The source is connected to the sink via direct arcs or via other SCCs.

\begin{figure}[t!]
    \centering
    \includegraphics[trim={0 6 0 10},clip,scale=0.9]{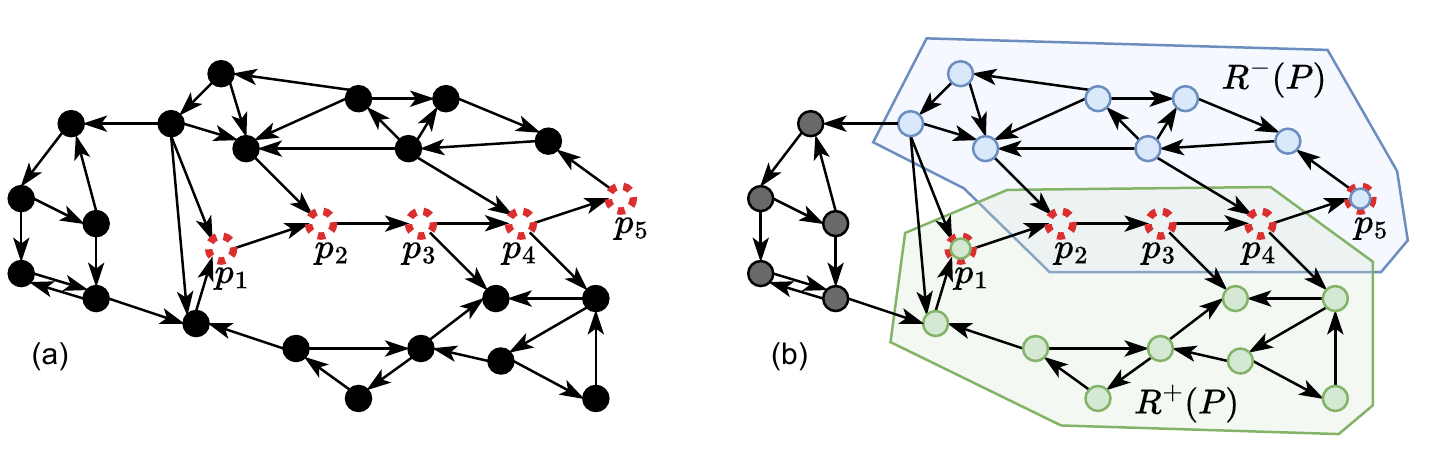}
    \caption{(a) A cut path $P = (p_1, p_2, p_3, p_4, p_5)$, highlighted by red dashed nodes.
    Each walk from $p_1$ to $p_5$ has $P$ as a subwalk.
    (b) The remainder structure of $P$.
    The set $R^+(P)$ is enclosed by the green area, and the set $R^-(P)$ is enclosed by the blue area.
    Its source component $S^-(P)$ is highlighted by blue nodes, its sink component $S^+(P)$ is highlighted by green nodes and the inner component $\bar{S}(P)$ is highlighted by grey nodes.
    The path component $\bar{P}(P)$ is the intersection of $R^+(P)$ and $R^-(P)$.
    Note that, to get from $S^+(P)$ to $S^-(P)$, one needs to traverse $P$ completely.} 
    \label{fig:remainder-structure}
\end{figure}

A similar structure exists for cut paths, which we call their \emph{remainder structure}.
This structure is helpful both to efficiently check whether a walk is a cut path, and for applying cut paths to other problems.
We define the remainder structure of a cut path $P = (p_1, \dots, p_\ell)$ (where $p_i$ are nodes), which we denote $R(P) = (S^-(P), \bar{S}(P), S^+(P), \bar{P}(P))$, as follows.
Let $R^+(P)$ be the set of nodes reachable by walks starting from the first node of $P$ without using $(p_{\ell-1}, p_\ell)$.
Symmetrically, let $R^-(P)$ be the set of nodes reaching the last node of $P$ without using $(p_1, p_2)$.
With this definition, $R^+(P)$ and $R^-(P)$ may intersect, so we define the \emph{source component} $S^-(P) := R^-(P) \setminus R^+(P)$, the \emph{sink component} $S^+(P) := R^+(P) \setminus R^-(P)$, the \emph{inner component} $\bar{S}(P) := V \setminus (S^- \cup S^+)$, and the \emph{path component} $\bar{P}(P) := R^-(P) \cap R^+(P)$.
In the remainder structure $R(P)$, for each node $u$ in the subgraph induced by $S^-(P)$ and each node $v$ in the subgraph induced by $S^+(P)$, the walk $P$ is on every $u\text{-}v$ walk in $G$.
Computing the remainder structure is trivial given its definition, and it additionally allows for a very simple way to efficiently check if a walk is a cut path.

\begin{theorem}[restate = innerpath, name = Efficient verification of cut paths]
    \label{thm:path}
    Let $P = (p_1, \dots, p_\ell)$ be a walk of length $\ell-1 \geq 1$.
    $P$ is a cut path if and only if $\bar{P}(P) = \{p_2, \dots, p_{\ell-1}\}$ (specifically, for $\ell = 2$, $\bar{P}(P)$ is empty).
    This property can be verified in $O(m)$ time. 
\end{theorem}

\subsection{Applications of cut paths}
\label{sec:intro-applications}

\subparagraph{Background.}

We apply cut paths to improve several reachability-based results that have been given over the last years for ``safe walks'', motivated by the genome assembly problem in bioinformatics~\cite{tomescu2017safe,DBLP:journals/talg/CairoMART19,cairo2020macrotigs,acosta2018safe,Rahman27072022,nagarajan2009parametric}. We give here a minimal self-contained description, and refer the reader to these papers for motivation and applications. One can formulate the genome assembly problem as finding one closed node-covering walk (i.e., passing through every node at least once)\footnote{To be precise, \cite{tomescu2017safe,DBLP:journals/talg/CairoMART19,cairo2020macrotigs,acosta2018safe} focus mostly on the arc-covering case, where the closed walks have to pass through all \emph{arcs} at least once. In this paper we focus on the node-covering case, for two reasons: first, it has a more direct relation to cut paths, and second, the arc-covering case can be reduced to it in linear time by subdividing every arc (i.e., introducing a node in the middle of every arc).} in a given strongly connected graph built from the input sequencing data. Since such graphs may admit multiple such walks, one can define a \emph{safe walk} as one appearing in any closed node-covering walk of a strongly connected graph. Formally:

\begin{definition}[Safe walk~\cite{tomescu2017safe}]
    Given a graph $G$, a walk is \emph{safe} if it is a subwalk of each possible closed node-covering walk of $G$.
\end{definition}

We are interested in enumerating all \emph{maximal} safe walks, namely all those that are not a proper subwalk of another safe walk. These can be thought as representing the maximal correct (partial) answers to the genome assembly problem. In~\cite{tomescu2017safe} it is argued that some specific types of walks used by genome assembly programs are safe walks, and thus finding \emph{all} maximal safe walks can be relevant in practice, since it can lead to longer parts of the genome being reconstructed.
Similar problems have been previously studied without the covering constraint but instead by considering subwalks of all possible walks from a given node~$s$ to a given node~$t$~\cite{cairo2020safety}, or with the covering constraint set to \emph{cover exactly once} (i.e., closed Eulerian walks)~\cite{nagarajan2009parametric,Acosta:2022aa}.

Safe walks can be characterised as follows.
Let $(w_1,\dots,w_\ell)$ be a walk. A path from $w_i$ to $w_j$, with $1 < i \leq j < \ell$, with first arc different from $(w_{j},w_{j+1})$, and last arc different from $(w_{i-1},w_i)$, is called a \emph{forbidden path}. Tomescu and Medvedev~\cite{tomescu2017safe} proved that a walk is safe if and only if it has no forbidden path and all its arcs are cut arcs. For a walk made up only of cut arcs, such a forbidden path can be seen as a NO-certificate, since it testifies that the walk is not safe. Even though NO-certificates are usually harder to check, Cairo et al.~\cite{cairo2020macrotigs} showed that the absence of a forbidden path can be checked in $O(m)$ time, but using complex data structures. Moreover, all maximal walks without forbidden paths can still be enumerated in $O(mn)$ time~\cite{DBLP:journals/talg/CairoMART19}. By appropriately splitting such walks at non-cut arcs, and removing duplicates, also maximal safe walks can be enumerated in $O(mn)$ time.\footnote{This fact was not observed previously in the literature (recall that in this paper we are defining safety in terms of node-covering walks), but follows by standard techniques of removing duplicates using a suffix tree. For completeness, we explain this in \Cref{sec:enumerating-safe-walks,s:deduplication}.}

One can also consider another variant of the problem, where one needs to assemble an unknown number of genomes from a graph. As such, one can formulate the genome assembly problem as finding a node-covering \emph{set} of closed walks (i.e., such that every node appears in at least one walk in the set). In this setting, the notion of safety is adapted as follows:

\begin{definition}[Multi-safe walk~\cite{tomescu2017safe,acosta2018safe}]
    Given a graph $G$, a walk is \emph{multi-safe} if it is a subwalk of some walk in each possible node-covering set of proper closed walks of $G$.
\end{definition}

The theory around multi-safe walks is less developed, the only algorithmic result being by Obscura Acosta et al.~\cite{acosta2018safe}, who showed that all maximal multi-safe walks can be enumerated in $O(m^2+n^3)$ time. This algorithm is also based on forbidden paths, with some additional conditions. One reason behind this lack of overall progress around multi-safe walks can be due to the lack of a YES-certificate, which requires building new machinery from scratch.

\subparagraph{Safety-related previous works.} The idea of partial solutions common to all solutions to a problem is very natural and has appeared in several other contexts. For example, Costa~\cite{Costa1994143} studied \emph{persistent edges} belonging to all maximum matching of a bipartite graph, and Hammer et al.~\cite{doi:10.1137/0603052} studied \emph{persistent nodes} belonging to all maximum stable sets. Recently, Bumpus et al.~\cite{essential-vertices} studied \emph{$c$-essential vertices}, defined as those contained in all $c$-approximate solutions to e.g.~Odd Cycle Transversal and Directed Feedback Vertex Set problems. See also Table~1 in~\cite{essential-vertices} for algorithms detecting \emph{some} $c$-essential vertices for several other NP-hard problems. As opposed to the latter problems, in this paper we tackle polynomially solvable problems (computing closed node-covering walks is trivial), and thus their safe partial solutions admit rich structures which can be exploited in getting efficient algorithms enumerating \emph{all} of them. %

\subparagraph{Our results.} We show that cut paths and their remainder structure provide a \emph{flexible} technique to study both safe and multi-safe walks. For example, they can be used to derive natural YES-certificates for both types of walks.

\begin{figure}[tbh]
    \centering
    \includegraphics[trim={0 6 0 10},clip,scale=0.9]{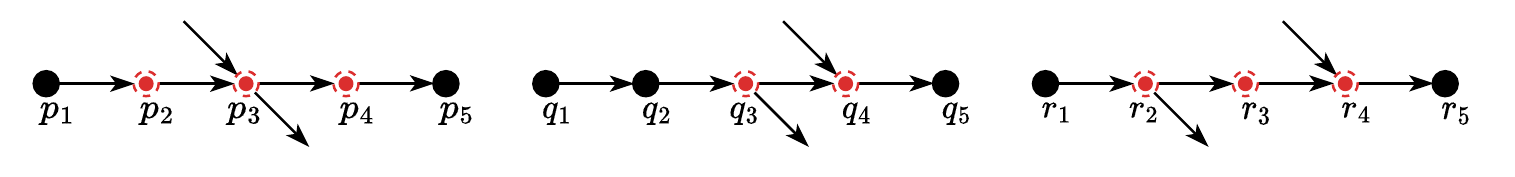}
    \caption{Walks $P$, $Q$, $R$ and their cores highlighted by red dashed nodes.
    Nodes $p_3$, $q_3$ and $r_2$ are splits, and nodes $p_3$, $q_4$ and $r_4$ are joins.
    The walks $P$ and $Q$ are interleaved, and the walk $R$ is non-interleaved.
    Note that $(q_1, \dots, q_5)$ is defined to be interleaved, since even though it has separate split-free and join-free parts, they are not trivially safe since $(q_3, q_4)$ is only safe if it is a cut arc.}
    \label{fig:walk-definitions}
\end{figure}

To describe our results, we need additional definitions for walks.
Examples for these definitions are given in \Cref{fig:walk-definitions}.
A \emph{split} is a node with at least two outgoing arcs and a \emph{join} is a node with at least two incoming arcs.
Let $W = (w_1, \dots, w_\ell)$ be a walk with $\ell \geq 2$.
The \emph{inner} nodes of $W$ are $w_2, \dots, w_{\ell-1}$.
Let $w_i$ be its first inner join, or $w_\ell$ if $W$ has no inner join.
Let $w_j$ be its last inner split, or $w_1$ if $W$ has no inner split.
Then $W$ is an \emph{interleaved walk} if $i \leq j + 1$ and a \emph{non-interleaved walk} otherwise.

The \emph{core} of an interleaved walk is its subwalk from $w_{i-1}$ to $w_{j+1}$.
The \emph{core} of a non-interleaved walk is its subwalk from $w_j$ to $w_i$.

A first consequence is that verifying whether a walk is safe can now be done by a simple check whether the \emph{core} of a walk is a cut path (after excluding trivial cases).
Since this can be computed in linear-time using simple graphs traversals (\Cref{thm:path}), we obtain a verification algorithm much simpler than the one in \cite{cairo2020macrotigs} (which uses complex data structures from~\cite{georgiadis2020strong}, which in turn uses dominator trees~\cite{fraczak2013finding} and loop-nesting forests~\cite{tarjan1976edge}).

\begin{theorem}[restate = onesafe, name = Safety characterisation and verification]
    \label{thm:safe}
    Let $W$ be a walk, and let $C(W)$ be its core.
    $W$ is safe if and only if it is a non-interleaved walk or $C(W)$ is a cut path.
    This property can be verified in $O(m)$ time.
\end{theorem}

In our proof, we use the properties of the remainder structure to show that if the core of a walk $W$ is not a cut path, then it can be replaced by a walk avoiding the core in any closed node-covering walk.
On the other hand, if the core is a cut path, then there is a pair of nodes that can only be connected via the core.
Since a closed node-covering walk contains a subwalk between each pair of nodes, that makes $W$ safe.

For the multi-safe case we obtain a YES-certificate based on checking a property of the remainder structure. This leads to the first linear-time algorithm verifying whether a walk is multi-safe.
In contrast to safe walks, the characterisation depends on the existence of certain SCCs of size one.
Intuitively, a multi-safe walk must be safe, since otherwise there would be a closed node-covering walk avoiding it, which then also disproves multi-safety.
Moreover, if $R^+(C(W))$, $R^-(C(W))$ and $\bar{S}(C(W))$ of a core $C(W)$ contain only SCCs of size at least two, then they can all be covered by proper closed walks without leaving the respective component, and thus without using $C(W)$ as subwalk, disproving the multi-safety of $W$.
If however one of them contains an SCC of size one, then to cover this SCC, the respective component needs to be left and reentered.
This can only happen by using $C(W)$ as subwalk, hence $W$ is multi-safe.

\begin{theorem}[restate = multisafe, name = Multi-safety characterisation and verification]
    \label{thm:multi-safe}
    Let $W$ be a walk, and let $C(W)$ be its core.
    If $W$ is non-interleaved, then it is multi-safe.
    Otherwise, it is multi-safe if and only if it is safe and any of $G[\bar{S}(C(W))], G[R^+(C(W))]$ or $G[R^-(C(W))]$ contains an SCC of size one.
    This property can be verified in $O(m)$ time.
\end{theorem}

Lastly, we improve the existing $O(m^2+n^3)$-time algorithm enumerating all maximal multi-safe walks. A naive application of \Cref{thm:multi-safe} would lead to an $O(m^2n)$-time algorithm for this enumeration problem, already improving the previous one for dense graphs.
However, by proving several additional properties of the remainder structure, %
we can amortise the time to just $O(mn + o)$, where $o$ is the size of the output.
First, the remainder structure of all subwalks $P'$ of a given cut path $P$ can be precomputed in linear time.
This works because when shifting either the start or end of the subwalk to the right (along the walk), then the set $R^+(P')$ monotonously grows and the set $R^-(P')$ monotonously shrinks (except for some subwalks that are trivial to handle without the remainder structure).
Further, when growing $R^+(P')$, only complete SCCs get removed from the inner component, and when shrinking $R^-(P')$, the existing SCCs in the inner component are not altered.

\begin{theorem}[restate = amortised, name = Enumerating maximal multi-safe walks]
    \label{thm:multi-safe-enum}
    All maximal multi-safe walks can be identified in $O(mn)$ time and enumerated in $O(mn + o)$ time, where $o$ is the total length of the output and it holds that $o \in O(n^3)$.
\end{theorem}

In a nutshell, using cut paths and the remainder structure, we gain a deeper understanding of critical structures for the connectivity of directed graphs.
In bioinformatics applications, this allows us to get better characterisations for two problems that for the first time admit a simple to compute and verify YES-certificate.
Moreover, meticulously investigating the properties of the remainder structure, we improve over the complexity of the best known enumeration algorithm for maximal multi-safe walks.

\subparagraph{Notation.} %
A \emph{split} is a node with at least two outgoing arcs, and a \emph{join} is a node with at least two incoming arcs.
Let $G = (V, E)$ be a strongly connected graph, with $|V| = n$ and $|E| = m \geq n$.
We denote the removal of an arc $e \in E$ by $G - e$.
Given a subset of nodes $V' \subseteq V$, the subgraph induced by $V'$ subgraph is defined as $G[V'] = (V', \{(u, v) \in E ~|~ u, v \in V'\})$.

For two nodes $u, v \in V$, a \emph{$u$-$v$ walk} of \emph{length} $\ell-1$ is a sequence of nodes $W = (w_1, \dots, w_\ell)$ with $w_1 = u$ and $w_\ell = v$ and such that for each $i \in \{1, \dots, \ell-1\}$ it holds that $(w_i, w_{i+1}) \in E$.
The \emph{tail} $\tail(W)$ of $W$ is $w_1$, and the \emph{head} $\head(W)$ of a $W$ is $w_\ell$.
$W$ is \emph{closed} if $u = v$ and \emph{open} otherwise.
$W$ is a \emph{path} if all nodes are unique except that $w_1 = w_\ell$ is allowed.
The \emph{inner nodes} of a $W$ are the nodes $w_2, \dots, w_{\ell-1}$, where a walk of length $\ell-1 \leq 1$ has no inner nodes.
The notation $WW'$ denotes the \emph{concatenation} of walks $W$ and $W'$ if $\head(W) = \tail(W')$.
\emph{Subwalks} of walks are defined in the standard manner, where subwalks of closed walks may run over the end (e.g. $(c, a, b)$ is a subwalk of $(a, b, c, a)$).
A \emph{proper subwalk} of $W$ is a subwalk that is shorter than $W$.

\section{Cut paths and their remainder structure}

In this section, we formally define cut paths and the remainder structure and prove their main properties.
See \Cref{s:proofs} for all formal proofs that we omitted here.

\begin{definition}
    \label{def:cut-path}
    A walk $W$ in a strongly connected graph $G = (V, E)$ is a \emph{cut path} if there is a pair of nodes $u, v \in V$ such that all $u$-$v$ walks in $G$ have $W$ as subwalk.
\end{definition}

One intuitive property of a cut path is that it is an open path.

\begin{lemma}[restate = openpath, name = Open path property]
    \label{lem:open-path}
    A cut path is an open path.
\end{lemma}

Further, for any cut path $P$, the pair of nodes $\tail(P), \head(P)$ is a \emph{witness} for $P$ being a cut path, i.e.~all $\tail(P)$-$\head(P)$ walks contain $P$ as a subwalk.

\begin{lemma}[restate = witness, name = Witness property]
    \label{lem:witness}
    A walk $W$ is a cut path if and only if it is subwalk of all $\tail(W)$-$\head(W)$ walks.
\end{lemma}

\subsection{Restricted reachabilities}

The remainder structure is based on the restricted reachabilities of a walk.
Even though only open paths can be cut paths, we define the remainder structure here on arbitrary walks.
As we see below, the remainder structure makes it easy to check if a walk is a cut path.

\begin{definition}
   \label{def:r}
    The \emph{restricted forward and backward reachability} of a walk $W = (w_1, \dots, w_\ell)$ in a strongly connected graph $G = (V, E)$ are
    \begin{align*}
        R^+(W) &:= \{v \in V \mid \exists\, \tail(W)\text{-}v \text{ walk in } G - (w_{\ell-1}, w_\ell)\},\\
        R^-(W) &:= \{v \in V \mid \exists\, v\text{-}\head(W) \text{ walk in } G - (w_1, w_2)\}.
    \end{align*}
\end{definition}

The restricted reachabilities exhibit the following property, which makes them simple to work with.

\begin{lemma}[Bottleneck property]
    \label{lem:bottleneck}
    For a walk $W = (w_1, \dots, w_\ell)$, the only arc leaving $R^+(W)$ is $(w_{\ell-1}, w_\ell)$ and the only arc entering $R^-(W)$ is $(w_1, w_2)$.
\end{lemma}
\begin{proof}
    Assume for a contradiction that an arc $(u, v)$ different from $(w_{\ell-1}, w_\ell)$ would leave $R^+(W)$.
    Since $u \in R^+(W)$, there is a $\tail(W)$-$u$ walk without $(w_{\ell-1}, w_\ell)$.
    Appending $(u, v)$ to such a walk cannot introduce $(w_{\ell-1}, w_\ell)$ as subwalk, because $(u, v)$ is not $(w_{\ell-1}, w_\ell)$.
    Therefore, $v \in R^+(W)$, contradicting $(u, v)$ leaving $R^+(W)$.
    
    By symmetry, the only arc entering $R^-(W)$ is $(w_1, w_2)$.
\end{proof}

Note that by leaving $R^+(W)$, one always ends up in $R^-(W)$ (or one was in $R^-(W)$ already, if one leaves $R^+(W)$ from a node in $R^+(W) \cap R^-(W)$), and by entering $R^-(W)$, one always comes from $R^+(W)$ (and possibly ends up in $R^+(W)$ again, if one enters $R^-(W)$ at a node in $R^+(W) \cap R^-(W)$).
Further, the restricted reachabilities are strongly connected in certain cases.

\begin{lemma}[Strong connectivity property]
    \label{lem:strong-connectivity}
    Let $P = (p_1, \dots, p_\ell)$ be a cut path.
    If the last inner node of $P$ is a split, then $G[R^+(P)]$ is strongly connected.
    If the first inner node of $P$ is a join, then $G[R^-(P)]$ is strongly connected.
\end{lemma}
\begin{proof}
    Let the last inner node of $P$ be a split.
    By definition, $p_1$ reaches all nodes in $R^+(P)$ via walks not leaving $R^+(P)$.
    Assume for a contradiction that there is a node $v \in R^+(P)$ that cannot reach $p_1$ without leaving $R^+(P)$.
    Then by \Cref{lem:bottleneck}, each $v$-$p_1$ walk contains $(p_{\ell-1}, p_\ell)$, so each $p_{\ell-1}$-$p_1$ walk contains $(p_{\ell-1}, p_\ell)$.
    Let $v' \neq v_\ell$ be a node with $(p_{\ell-1}, v') \in E$.
    Then since $v'$ is reachable from $p_{\ell-1}$, each $v'$-$p_1$ walk contains $(p_{\ell-1}, p_\ell)$.
    So there is a $p_1$-$p_\ell$ walk $W$ via $v'$ that does not have $p_1$ or $p_\ell$ as inner nodes, so it does not have $P$ as subwalk.
    By \Cref{lem:witness}, this contradicts $P$ being a cut path.
\end{proof}

Note that, whenever a restricted reachability is not strongly connected, then it consists of a strongly connected component, plus nodes from $P$ that form SCCs of size one.

\subsection{The remainder structure}

\begin{definition}
    \label{def:remainder-structure}
    The \emph{remainder structure} $R(W) = (S^-(W), \bar{S}(W), S^+(W), \bar{P}(W))$ of a walk $W$ in a strongly connected graph $G = (V, E)$ is defined as
    \begin{align*}
        S^-(W) &:= R^-(W) \setminus R^+(W) \text{ \emph{(the source component)}},\\
        \bar{S}(W) &:= V \setminus (R^+(W) \cup R^-(W)) \text{ \emph{(the inner component)}},\\
        S^+(W) &:= R^+(W) \setminus R^-(W) \text{ \emph{(the sink component)}},\\
        \bar{P}(W) &:= R^+(W) \cap R^-(W) \text{ \emph{(the path component)}}.
    \end{align*}
\end{definition}

Note that the remainder structure is a decomposition of the nodes of $G$.
For checking if a walk is a cut path, we can use the \emph{inner path property} of the remainder structure.
The inner path property can be checked in linear time with trivial algorithms that directly follow from the definition of the remainder structure and the property.

\innerpath*
\begin{proof}
    By definition, it holds that $\{p_2, \dots, p_{\ell-1}\} \subseteq \bar{P}(P)$.
    Assume for a contradiction that there was a node $v \in \bar{P}(P) \setminus \{p_2, \dots, p_{\ell-1}\}$.
    If $v = p_\ell$, then $p_\ell \in R^+(P)$, so there is a $\tail(P)\text{-}\head(P)$ walk that does contain $(p_{\ell-1}, p_\ell)$, which by \Cref{lem:witness} contradicts $P$ being a cut path.
    In the same way, if $v = p_1$, then $p_1 \in R^-(P)$, which again contradicts $P$ being a cut path.
    Therefore, $v \notin \{p_1, \dots, p_\ell\}$.
    
    Since $v \in R^+(P)$, there is a $\tail(P)\text{-}v$ walk $W_1$ in $G$ that does not contain $(p_{\ell-1}, p_\ell)$.
    Further, since $v \in R^-(P)$, there is a $v\text{-}\head(P)$ walk $W_2$ in $G$ that does not contain $(p_1, p_2)$.
    Then, $W = W_1W_2$ is a $\tail(P)\text{-}\head(P)$ walk.
    Since $v \notin \{p_1, \dots, p_\ell\}$, it holds that concatenating $W_1W_2$ does not introduce $P$ as subwalk.
    So by \Cref{lem:witness} it holds that $P$ is not a cut path, which completes the contradiction.
    
    If $\ell > m$, then $W$ contains a cycle, so by \Cref{lem:open-path} it is not a cut path.
    The sets $R^+(W)$ and $R^-(W)$ can be computed in linear time and hence $R^+(W) \cap R^-(W)$ can be computed in linear time.
    Resulting, $P$ being a cut path can be verified in $O(m)$ time.
\end{proof}

The remainder structure exhibits two more properties useful for other problems.

\begin{lemma}[restate = extendedwitness, name = Extended witness property]
    \label{lem:extended-witness}
    Let $P = (p_1, \dots, p_\ell)$ a cut path of length $\ell-1 \geq 1$.
    Let $u \in S^+(P)$ and $v \in S^-(P)$ be nodes.
    It holds that all $u$-$v$ walks contain $P$ as subwalk.
\end{lemma}

\begin{lemma}[restate = nonemptiness, name = Nonemptiness property]
    \label{lem:nonemptiness}
    For a cut path $P = (p_1, \dots, p_\ell)$ of length $\ell-1 \geq 1$, it holds that $p_1 \in S^+(P)$ and $p_\ell \in S^-(P)$.
\end{lemma}

\section{Linear-time verifiable characterisations of (multi-)safe walks}

We apply the remainder structure of a cut path to give easily and efficiently verifiable characterisations of safe and multi-safe walks.
From here on we assume that our strongly connected graph $G = (V, E)$ is not a cycle.
All missing formal proofs are in \Cref{s:proofs}.

\begin{sloppypar}
First note that the univocal extension of a walk always needs to be traversed when traversing the walk itself with a closed walk.
The \emph{univocal extension} $U(W) = (l_1, \dots, l_a, w_1, \dots, w_\ell, r_1, \dots, r_b)$ of $W$ is a maximal walk where $l_2, \dots, l_a, w_1$ are not joins and $w_\ell, r_1, \dots, r_{b-1}$ are not splits.
\end{sloppypar}

\begin{lemma}[restate = univocalextensionsafety, name = Safety of univocal extensions]
    \label{lem:univocal-extension-safety}
    The univocal extension $U(W)$ of a walk $W$ is safe if and only if $W$ is safe.
    It is multi-safe if and only if $W$ is multi-safe.
\end{lemma}

Since walks are (not necessarily maximal) univocal extensions of their cores, we get the following property useful for characterising safe and multi-safe walks.

\begin{lemma}[restate = coresafety, name = Safety of cores]
    \label{lem:core-safety}
    The core $C(W)$ of a walk $W$ is safe if and only if $W$ is safe.
    It is multi-safe if and only if $W$ is multi-safe.
\end{lemma}

The difficulty in characterising safe and multi-safe walks lies in characterising interleaved walks.
Non-interleaved walks are very simple to handle.

\begin{lemma}[restate = noninterleaved, name = Non-interleaved walks]
    \label{lem:non-interleaved}
    A non-interleaved walk $W = (w_1, \dots, w_\ell)$ with core $C(W) = (w_j, \dots, w_i)$, $j + 2 \leq i$, is both safe and multi-safe.
    This property can be verified in $O(m)$ time.
\end{lemma}

Safe walks can be characterised as follows.

\onesafe*
\begin{proof}
    If $W$ is non-interleaved, the statement follows by \Cref{lem:non-interleaved}.

    If the core $C(W)$ of $W$ is not a cut path, then by \Cref{lem:witness}, there is a $\tail(C(W))$-$\head(C(W))$ walk that does not have $C(W)$ as subwalk.
    This can be used to replace all occurrences of $C(W)$ in a node-covering closed walk.
    To ensure that the resulting walk $C$ covers all nodes, we insert two closed walks $C_1$ and $C_2$ into it, constructed as follows.
    Let $v$ be the last inner split in $C(W)$ and $v'$ one of its successors outside of $C(W)$.
    $C_1$ starts in $\tail(C(W))$ and walks $C(W)$ until $v$ and then $v'$.
    From $v'$ it walks back to $\tail(C(W))$, which is possible because the graph is strongly connected, and also possible without $C(W)$ as subwalk since it ends in $\tail(C(W))$.
    $C_2$ is constructed symmetrically through $\head(C(W))$.
    Since $W$ is interleaved, it holds that $C(W)$ is interleaved, so $C_1$ and $C_2$ together cover $C(W)$.
    Further, because $C$ contains $\tail(C(W))$ and $\head(C(W))$, we can insert $C_1$ and $C_2$ into it.
    Since the insertions happen at the first/last node of $C(W)$ they do not introduce it as subwalk.
    Thus, $W$ is not safe.
    
    If the core $C(W)$ of $W$ is a cut path, then there is a pair of nodes $u, v \in V$ such that all $u$-$v$ walks have $C(W)$ as a subwalk.
    A closed node-covering walk contains a subwalk between each pair of nodes, so also between $u$ and $v$.
    Therefore, each closed node-covering walk contains $C(W)$ as subwalk, so by \Cref{lem:core-safety}, $W$ is safe.
    
    Finally, the core $C(W)$ can be computed in linear time, and by \Cref{thm:path} it can be checked for being a cut path in linear time.
    Hence, verifying if $W$ is safe takes $O(m)$ time.
\end{proof}

\begin{figure}[tbh]
    \centering
    \includegraphics[trim={0 6 0 10},clip,scale=0.9]{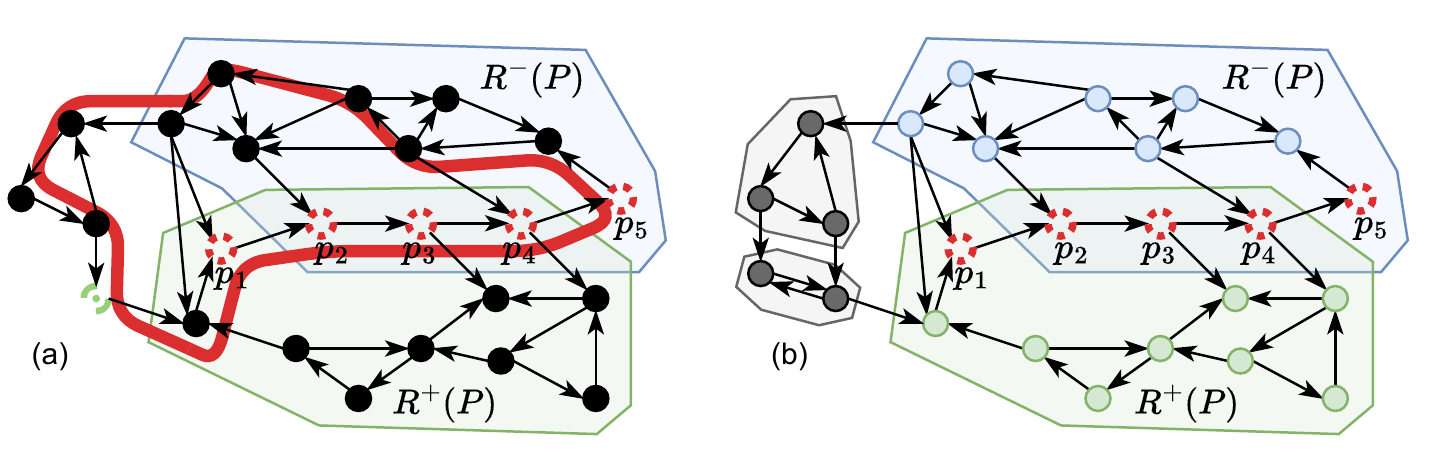}
    \caption{(a) A walk $(p_1,p_2,p_3,p_4,p_5)$ (as red dashed nodes) that is multi-safe, because the inner component of its remainder structure contains an SCC of size one (the green dashed node).
    The red walk is an example proper closed walk that covers the marked SCC.
    (b) Neither the restricted reachabilities nor the inner component contain an SCC of size one, so the two SCCs in the inner component and the restricted reachabilities can be covered by separate closed walks, without traversing the walk.}
    \label{fig:multi-safe}
\end{figure}

Multi-safe walks can be characterised as follows.
See \Cref{fig:multi-safe} for an example.

\multisafe*
\begin{proof}
    If $W$ is non-interleaved, the statement follows by \Cref{lem:non-interleaved}.
    
    If $C(W)$ is not a cut path, then by \Cref{thm:safe}, $W$ is not safe, so it is also not multi-safe.
    
    If none of $G[\bar{S}(C(W))], G[R^+(C(W))]$ and $G[R^-(C(W))]$ contain an SCC of size one, then all nodes can be covered by proper closed walks that do not leave a single one of $G[\bar{S}(C(W))], G[R^+(C(W))]$ or $G[R^-(C(W))]$.
    By definition, $G[\bar{S}(C(W))]$ contains no node of $C(W)$.
    Further, by \Cref{lem:bottleneck}, $G[R^+(C(W))]$ misses the last node of $C(W)$ and $G[R^-(C(W))]$ misses the first node of $C(W)$.
    Thus, none of the proper closed walks can have $C(W)$ as subwalk, so $W$ is not multi-safe.
    
    If $C(W)$ is a cut path and $G[\bar{S}(C(W))]$ contains an SCC of size one, then this SCC cannot be covered by a proper closed walk without leaving $\bar{S}(C(W))$.
    By \Cref{lem:bottleneck,lem:nonemptiness}, when leaving $\bar{S}(C(W))$, a walk ends up in $S^+(C(W))$, and when entering $\bar{S}(C(W))$, a walk comes from $S^-(C(W))$.
    By \Cref{lem:extended-witness}, walking from $S^+(C(W))$ to $S^-(C(W))$ requires having $C(W)$ as a subwalk.
    Therefore, by \Cref{lem:core-safety}, $W$ is multi-safe.
    
    If $C(W)$ is a cut path and $G[R^+(C(W))]$ contains an SCC of size one, then by \Cref{lem:strong-connectivity}, $R^+(C(W))$ contains exactly one node $v$.
    By definition, $v$ is not contained in $G[\bar{S}(C(W))]$ and by \Cref{lem:nonemptiness}, $v$ is not contained in $G[R^-(C(W))]$.
    So to cover $v$, a proper closed walk $X$ would leave $R^+(C(W))$.
    By \Cref{lem:bottleneck,lem:nonemptiness}, this implies that $X$ would have a subwalk from $\tail(C(W))$ to $\head(C(W))$.
    By \Cref{lem:witness} this implies $X$ has $C(W)$ as subwalk, so by \Cref{lem:core-safety}, $W$ is multi-safe.
    By symmetry, if $C(W)$ is a cut path and $G[R^-(C(W))]$ contains an SCC of size one, then $W$ is multi-safe.
    
    Finally, the core $C(W)$ can be computed in linear time, and by \Cref{thm:path} it can be checked for being a cut path in linear time.
    Moreover, the strongly connected components can be computed in linear time by~\cite{tarjan1972depth}.
    Hence, $W$ being safe can be verified in $O(m)$ time.
\end{proof}

\section{Amortised enumeration of all maximal multi-safe walks}

In this section we give the algorithm supporting \Cref{thm:multi-safe-enum}. Since every multi-safe walk is also safe, by definition, we start by enumerating all safe walks, and then finding their subwalks that are also multi-safe, using further properties of their remainder structure, including an additional monotonicity property of it.

\subparagraph{Enumerating all maximal safe walks}
\label{sec:enumerating-safe-walks}

Recall that all walks without forbidden paths can be enumerated in time $O(mn)$ with the algorithm from Cairo et al.~\cite{DBLP:journals/talg/CairoMART19}. From these, it is simple to get the safe walks using the following property from~\cite[Theorem 3]{tomescu2017safe}:

\begin{lemma}[Safe walks~\cite{tomescu2017safe}]
    \label{lem:arc-to-node}
    A walk is safe if and only if it has no forbidden paths and has no cut arc.
\end{lemma}

We compute the cut arcs in linear time~\cite{italiano2012finding} and then break all the walks without forbidden paths at arcs that are not cut arcs.
Then we remove duplicates and proper subwalks from the result using standard methods and suffix trees. See \Cref{s:deduplication} for details.

\begin{lemma}[Enumeration of safe walks]
    \label{thm:omnitig-enumeration}
    All maximal safe walks can be enumerated in $O(mn)$ time.
\end{lemma}

\subparagraph{Enumerating all maximal multi-safe walks}
\label{sec:enumerating-multi-safe-walks}

\begin{algorithm}[t!]
    \caption{\textsc{MultiSafe}}
    \label{alg:multitigs}
    \KwIn{Strongly connected graph $G = (V, E)$, all maximal safe walks $\mathcal{W}$.}
    \KwOut{All maximal multi-safe walks $\mathcal{W}'$.}
    \DontPrintSemicolon
    \vspace{0.2em}
    
    $\mathcal{W}' \gets ()$ \tcp{empty list}
    \For{$W \in \mathcal{W}$\label{alg:multitigs:outer-loop}}{
        \If{$W$ is a non-interleaved walk}{
            append $W$ to $\mathcal{W}'$, \Continue\;
        }
    
        $(w_1, \dots, w_\ell) \gets C(W)$, $start \gets 1$, $end \gets 1$ \;
        \While{$end \leq \ell$}{
            \If{$(w_{start}, \dots, w_{end})$ is multi-safe}{
                $end \gets end + 1$ \;
            }\Else{
                append $U((w_{start}, \dots, w_{end-1}))$ to $\mathcal{W}'$ \;
                $start \gets start + 1$ \;
            }
        }
        \While{$(w_{start}, \dots, w_{end-1})$ is not multi-safe}{
            $start \gets start + 1$ \;
        }
        append $U((w_{start}, \dots, w_{end-1}))$ to $\mathcal{W}'$ \;
    }
    
    Remove duplicates and subwalks from $\mathcal{W}'$ using e.g.~a suffix tree \;
\end{algorithm}

Using \Cref{thm:multi-safe}, we are able to derive an algorithm that enumerates maximal multi-safe walks.
It works by iterating all safe walks and enumerating all their maximal multi-safe subwalks.
We start with the subwalk of a safe walk consisting of its first core arc, and then extend it to the right whenever it is safe, while removing its first node (only from the subwalk, not the graph) whenever it is not safe.
Afterwards, we deduplicate and remove proper subwalks again, as described in \Cref{s:deduplication}.
See \Cref{alg:multitigs} for pseudocode.
To analyse the runtime of this algorithm without amortisation, we use the following property.

\begin{lemma}[restate = coreinterleavedamount, name = Amount of interleaved safe walks]
    \label{lem:core-interleaved-amount}
    A strongly connected graph contains at most $O(n)$ interleaved maximal safe walks.
\end{lemma}

By Cairo et al.~\cite{DBLP:journals/talg/CairoMART19}, the total length of the maximal walks without forbidden paths is $O(mn)$, and hence the total length of the maximal safe walks is $O(mn)$.
Thus, if there are no interleaved walks, the algorithm runs in $O(mn)$ time.
However, there may be up to $O(n)$ interleaved walks, and since their cores can not have cycles, for each interleaved walk, \Cref{alg:multitigs} performs up to $O(n)$ multi-safety checks.
Each such check takes $O(m)$ by \Cref{thm:multi-safe}.
Further, in a strongly connected graph that is not a cycle, a univocal extension increases the length of a walk by at most $O(n)$.
So the total length of the maximal multi-safe walks produced by a maximal safe walk is $O(n^2)$.
Hence, including linear-time deduplication and removal of proper subwalks, the runtime of \Cref{alg:multitigs} is $O(mn)$ for non-interleaved walks plus $O(mn + n^2)$ for each of the up to $O(n)$ interleaved walks.
Summed up, that is $O(mn^2 + n^3) = O(mn^2)$.
However, using amortisation, we get $O(mn)$, as shown below.

\begin{theorem}[Amortised computation]
    \label{thm:amortised}
    Let $P = (p_1, \dots, p_\ell)$ be a cut path and let $P' = (p_i, \dots, p_j)$ be a subwalk of $P$ with $i < j$.
    If $p_{j-1}$ and $p_{\ell-1}$ are splits, then $R^+(P') \subseteq R^+(P)$.
    If $p_2$ and $p_{i+1}$ are joins, then $R^-(P') \subseteq R^-(P)$.
\end{theorem}
\begin{proof}
    Let $p_{j-1}$ and $p_{\ell-1}$ be splits.
    Let $v \notin R^+(P)$.
    Then by definition each $p_1$-$v$ walk contains $(p_{\ell-1}, p_\ell)$.
    Further, since $i \neq \ell$ and by \Cref{lem:open-path} it holds that $P$ is an open path, it holds that $p_i \in R^+(P)$.
    So each $p_i$-$v$ walk contains $(p_{\ell-1}, p_\ell)$.
    Further, since $P$ is a cut path, also its subwalk $P_C = (p_i, \dots, p_{\ell-1})$ is a cut path.
    By \Cref{lem:witness}, this implies that each $p_i$-$p_{\ell-1}$ walk has $P_C$ as subwalk, which especially means that it contains $(p_{j-1}, p_j)$.
    Therefore, each $p_i$-$v$ walk contains $(p_{j-1}, p_j)$, so $v \notin R^+(P')$.
    Resulting, $R^+(P') \subseteq R^+(P)$.
    
    By symmetry, if $p_2$ and $p_{i+1}$ are joins, then $R^-(P') \subseteq R^-(P)$.
\end{proof}

\begin{figure}[tbh]
    \centering
    \includegraphics[trim={0 10 0 9},clip,scale=0.75]{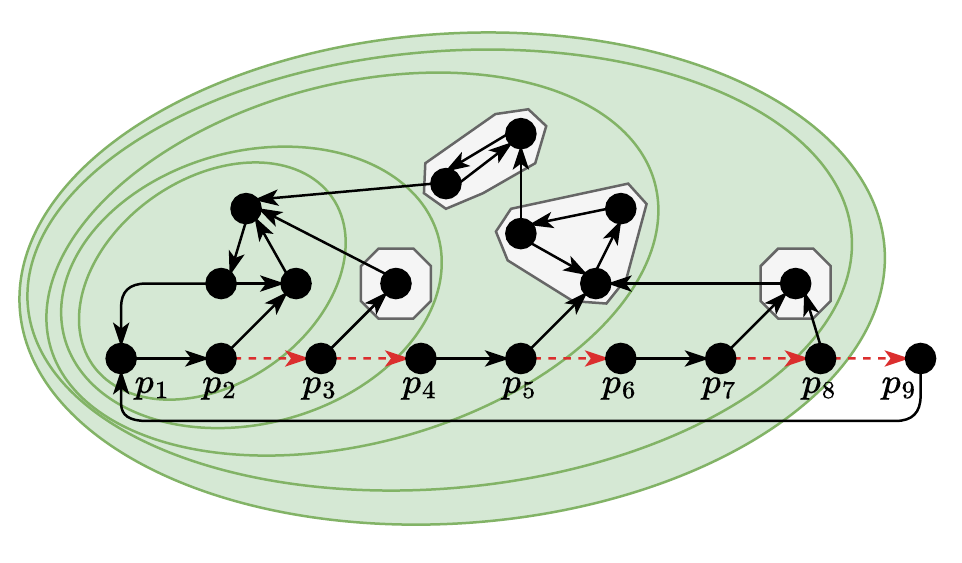}
    \caption{Amortised computation of $R^+(P')$ where $P'$ is a prefix of $P = (p_1, \dots, p_9)$.
    The red dashed arcs mark the ends of the prefixes $P'$.
    They are also the only arc leaving their respective $R^+(P')$.
    The SCCs of $G \setminus R^+(P')$ are enclosed in grey areas, and the different $R^+(P')$ are enclosed in green areas.}
    \label{fig:amortised}
\end{figure}

An example for \Cref{thm:amortised} is given in \Cref{fig:amortised}.
Using \Cref{thm:amortised} to implement \Cref{alg:multitigs}, we can answer all multi-safety queries for a single walk $W \in \mathcal{W}$ in $O(m)$ time.
For this, we observe that the boundaries of the subwalk $W'$ of $C(W)$ that is tested for safety only get shifted towards the end of $C(W)$.
Further, whenever $W'$ contains no split or no join as inner node, then it is either a non-interleaved walk or a univocal extension of a cut arc (by \Cref{lem:arc-to-node}) and hence safe.
So we only need to compute the remainder structure for subwalks that contain at least one split and one join.
For such subwalks, by \Cref{thm:amortised} it holds that $R^+(W')$ is monotonically increasing within each execution of the body of the loop in \Cref{alg:multitigs:outer-loop}, and $R^-(W')$ is monotonically decreasing.
Therefore, all $R^+(W')$ and $R^-(W')$ can be precomputed by computing them for all prefixes and suffixes of $W$ that contain splits or joins, respectively.
The computation of the $R^+(W')$ is done in forward order, and for each node, the search is started from the node itself and nodes visited by earlier searches are pruned.
With this strategy, the $R^+(W')$ of all prefixes $W'$ of $C(W)$ can be computed in $O(m)$ time.
By computing the $R^-(W')$ in reverse order, they can also be computed in $O(m)$ time.

To check the safety of all subwalks $W'$ based on the precomputed remainder structure in linear time, note the following.
When executing the body of the loop in \Cref{alg:multitigs:outer-loop}, $R^+(W')$ only increases and $R^-(W')$ only decreases for the relevant subwalks (those that are interleaved).
Further, by \Cref{lem:strong-connectivity}, $G[R^+(W')]$ and $G[R^-(W')]$ are strongly connected if they have a split as last inner node or a join as first inner node, respectively.
If they are not strongly connected then the inner nodes after the last inner split or the inner nodes before the first inner join, respectively, form SCCs of size 1.
Therefore, the query if $G[R^+(W')]$ or $G[R^-(W')]$ contain an SCC of size 1 can be answered in constant time if the size of $R^+(W')$ and $R^-(W')$ as well as the joins and splits of $W'$ are tracked within the body of the loop in \Cref{alg:multitigs:outer-loop}.
Further, by definition of $R^+(W')$ and $R^-(W')$, when $R^+(W')$ grows, SCCs of the inner component enter it as a whole, so by growing $R^+(W')$, the remaining SCCs of the inner component remain unchanged.
Symmetrically, when shrinking $R^-(W')$, the new nodes do not alter the existing SCCs in the inner component, so to check whether an SCC of size one was added, only the SCCs of the induced subgraph of the newly added nodes need to be computed.
Since the SCCs of a graph can be computed in linear time~\cite{tarjan1972depth}, this means that tracking whether there are SCCs of size 1 in any of $G[\bar{S}(W')]$, $G[R^+(W')]$ or $G[R^-(W')]$ can be implemented in $O(m)$ time per execution of the body of the loop in \Cref{alg:multitigs:outer-loop}.

Hence, we get a runtime of $O(m)$ for the multi-safety checks of each of the $O(n)$ interleaved safe walks.
However, each interleaved maximal safe walk may produce maximal multi-safe walks of total length $O(n^2)$.
But, amortising over all interleaved maximal safe walk, we see that each core of a multi-safe walk is uniquely identified by its first and last node, since otherwise it would not be a cut path by \Cref{lem:witness}.
So, we can use a flag for each pair of nodes and thus avoid repetitions in multi-safe walks produced by interleaved maximal safe walks in constant time.
In total, at most $O(n^2)$ such checks happen, so the interleaved walks take $O(mn + o)$ time, where $o$ is the total length of the interleaved walks.
This results in a total time of $O(mn + o)$ for \Cref{alg:multitigs}, and since there are at most $n$ interleaved safe walks by \Cref{lem:core-interleaved-amount}, $o \in O(n^3)$.
By reporting the maximal multi-safe walks as start and end index in their respective maximal safe walks, we get an output size of $O(n^2)$ interleaved maximal multi-safe walks, plus $O(mn)$ non-interleaved maximal multi-safe walks.
So if we are only interested in identifying maximal multi-safe walks and not in an explicit enumeration, we have an algorithm that runs in $O(mn)$ time.

\amortised*

\section{Conclusions and future work}

We introduced cut paths as a generalisation of cut arcs, as well as the remainder structure of cut paths. 
Using properties of the remainder structure, we applied cut paths to some well-studied reachability problems from bioinformatics.
In the same way as the remainder structure gave a simple YES-certificate for a path to be a cut path (\Cref{thm:path}), the remainder structure led to easily verifiable YES-certificates for walk safety (\Cref{thm:safe}) and multi-safety (\Cref{thm:multi-safe}), which were open questions. By proving an additional monotonicity property (\Cref{thm:amortised}), we improved the state-of-the-art of enumeration of all maximal multi-safe walks (\Cref{thm:multi-safe-enum}).

There are central structural questions about cut paths that remain open.
It is known that there are at most $O(n)$ cut arcs which can be enumerated in $O(m)$ time~\cite{italiano2012finding}.
But for cut paths, there is no known upper bound to their amount or total length, and it is open how they can be enumerated efficiently.
Further, it is open how they can overlap and intersect.

For our applications, it is open if the total length of safe and multi-safe walks is really $O(mn)$, or if our enumeration algorithm for safe and multi-safe walks is not optimal.
Further, it is open if there is a linear output-sensitive algorithm for safe or multi-safe walks, as the one for walks without forbidden paths from~\cite{cairo2020macrotigs}.

\bibliography{bibliography}

\appendix

\section{Deduplication and removal of proper subwalks in linear time}
\label{s:deduplication}

\begin{algorithm}[tbh]
    \caption{\textsc{RemoveDuplicatesAndProperSubwalks}}
    \label{alg:deduplication}
    \KwIn{List of walks $\mathcal{W} = (W_1, \dots, W_{|\mathcal{W}|})$.}
    \KwOut{Set of walks $\mathcal{W'}$ containing one copy of each unique walk in $\mathcal{W}$ that is not a proper subwalk of another walk in $\mathcal{W}$.}
    \DontPrintSemicolon
    \vspace{0.2em}
    
    Sort $\mathcal{W}$ by length descending \;
    Build string $S = W_1\$W_2\$ \dots \$W_{|\mathcal{W}|}\$$ \;
    Build suffix tree $T$ on $S$ \;
    
    $\mathcal{W} \gets \emptyset$ \;
    \For{$W_i \in \mathcal{W}$}{
        $(l, r) \gets$ first occurrence of $W_i$ in $S$ \;
        \If{$(l, r) = $ coordinates of $W_i$ in $S$}{
            $\mathcal{W'} \gets \mathcal{W'} \cup \{W_i\}$ \;
        }
    }
\end{algorithm}

The removal of duplicates and subwalks is implemented in linear time using a suffix tree in \Cref{alg:deduplication}.
It works by sorting the walks by length descending and only reporting a walk if it is no subwalk of a previous walk, meaning if it is no proper subwalk of a previous walk and it is the first occurrence of itself.

\begin{lemma}[Deduplication]
    \label{lem:deduplication}
    \Cref{alg:deduplication} is correct and works in time linear in the total length of $\mathcal{W}$.
\end{lemma}
\begin{proof}
    The algorithm sorts the walks by length descending and then reports walks only if their string of nodes does not occur any earlier than themselves in $S$.
    This reports only one copy of each walk since only the first occurrence of a walk in $S$ is reported.
    Moreover, if a walk is a proper subwalk of another, then that subwalk will occur earlier in $S$, so the subwalk will never be reported.
    Therefore, \Cref{alg:deduplication} is correct.
    
    For the runtime, let $||\mathcal{W}||$ be the total length of $\mathcal{W}$.
    Sorting $\mathcal{W}$ by length descending can be done by bucket sort with $||\mathcal{W}||$ buckets containing each a dynamic array.
    Then it runs in time linear in $||\mathcal{W}||$.
    Building $S$ and the suffix tree is linear in $||\mathcal{W}||$~\cite{ukkonen1995line}.
    The total cost of checking for the first occurrences of all $W_i$ in $S$ is linear in $||\mathcal{W}||$~\cite{gusfield1997algorithms}.
    Checking the coordinates of $W_i$ can be done by storing the coordinates for each string while constructing $S$ in a lookup table indexed by $i$.
    Then the branch runs in constant time.
    Therefore, \Cref{alg:deduplication} runs in time linear in $||\mathcal{W}||$.
\end{proof}

\section{Omitted proofs}
\label{s:proofs}

This section contains the proofs omitted from the main matter.

\openpath*
\begin{proof}
    Let $W$ be a walk in a strongly connected graph $G = (V, E)$ that is not an open path, i.e. it repeats some node $v$.
    Then we can construct the walk $W' \neq W$ by removing all $v$-$v$ subwalks from $W$ (if $W$ is a $v$-$v$ walk, then $W'$ is a single node).
    Assume for a contradiction that $W$ was a cut path.
    Then there would be a pair of nodes $u, w \in V$ such that every $u$-$w$ walk in $G$ would have $W$ as subwalk.
    But we could replace all occurrences of $W$ by $W'$ in any $u$-$w$ walk, resulting in $u$-$w$ walks that do not have $W$ as subwalk.
    By contradiction, $W$ is not a cut path.
\end{proof}

\witness*
\begin{proof}
    Let $G = (V, E)$ be the strongly connected graph that contains $W$.
    If there is a $\tail(W)$-$\head(W)$ walk without $W$ as subwalk, then for any pair of nodes $u, v \in V$, any $u$-$v$ walk that contains $W$ as subwalk can be transformed into a $u$-$v$ walk that does not contain $W$ as subwalk.
    Then $W$ is not a cut path.
    
    If all $\tail(W)$-$\head(W)$ walks have $W$ as subwalk, then by definition $W$ is a cut path.
\end{proof}

\extendedwitness*
\begin{proof}
    Let $W_1W_2W_3$ be a $\tail(P)$-$\head(P)$ walk where $W_1$ is a $\tail(P)$-$u$ walk, $W_2$ is a $u$-$v$ walk, and $W_3$ is a $v$-$\head(P)$ walk.
    By \Cref{lem:witness}, $W_1W_2W_3$ must have $P$ as a subwalk since $P$ is a cut path.
    By definition, $\tail(P)$ reaches all nodes in $S^+(P) \subseteq R^+(P)$ without using $(p_{\ell-1}, p_\ell)$, so we can choose $W_1$ without using $P$ as subwalk.
    Also, all nodes in $S^-(P) \subseteq R^-(P)$ reach $\head(P)$ without using $(p_1, p_2)$, so we can choose $W_3$ without using $P$ as subwalk.
    Further, by definition, $u, v \notin \bar{P}(P)$, so by \Cref{thm:path}, neither $u$ nor $v$ are inner nodes of $P$.
    Therefore, concatenating $W_1W_2W_3$ cannot introduce $P$ as a subwalk by crossing the boundary between either $W_1$ and $W_2$ or $W_2$ and $W_3$.
    Concluding, $W_2$ has $P$ as subwalk.
\end{proof}

\nonemptiness*
\begin{proof}
    By definition, $p_1 \in R^+(P)$ and $p_\ell \in R^-(P)$.
    It holds that $p_\ell \notin R^+(P)$ and $p_1 \notin R^-(P)$, since by \Cref{lem:bottleneck} any of the two implies a $\tail(P)$-$\head(P)$ walk without $P$ as subwalk, which by \Cref{lem:witness} contradicts $P$ being a cut path.
\end{proof}

\univocalextensionsafety*
\begin{proof}
    Let $U(W) = (u_1, \dots, u_\ell)$ and $W = (u_i, \dots, u_j)$.
    Any closed walk having $W$ as subwalk can only enter $W$ via $(u_1, \dots, u_i)$ since $(u_2, \dots u_i)$ has no joins, and it can only leave $W$ via $(u_j, \dots, u_\ell)$ since $(u_j, \dots, u_{\ell-1})$ has no splits.
    Hence, $U(W)$ is safe if and only if $W$ is safe, and $U(W)$ is multi-safe if and only if $W$ is multi-safe.
\end{proof}

Since walks are (not necessarily maximal) univocal extensions of their cores, we get the following property useful for characterising safe and multi-safe walks.

\coresafety*
\begin{proof}
    Let $W = (w_1, \dots, w_\ell)$ and $C(W) = (w_i, \dots, w_j)$.
    By definition, independent of $W$ being interleaved or non-interleaved, it holds that $(w_2, \dots, w_i)$ contains no joins and $(w_j, \dots, w_{\ell-1})$ contains no splits.
    Hence, $W$ is subwalk of $U(C(W))$, so $U(W) = U(C(W))$.
    By \Cref{lem:univocal-extension-safety}, $W$ is safe $\iff$ $U(W)$ is safe $\iff$ $U(C(W))$ is safe $\iff$ $C(W)$ is safe.
    The same equivalence holds for the multi-safe property.
\end{proof}

\noninterleaved*
\begin{proof}
    Since $j + 2 \leq i$, it holds that $C(W)$ has an inner node that can only be covered by a closed walk by using $C(W)$ as subwalk.
    Hence, by \Cref{lem:core-safety}, it holds that $W$ is both safe and multi-safe.
    
    Finally, checking a walk for being interleaved can be done in $O(m)$ time.
    If $\ell > 3m$, then if the graph is a cycle, $W$ is non-interleaved.
    If the graph is not a cycle, then assume for a contradiction that $W$ is non-interleaved.
    Then $W$ contains a join-free or split-free subwalk of length $m+1$.
    Such a subwalk contains a cycle, because $m+1 > m$.
    And such a cycle is then join-free or split-free.
    This contradicts the graph not being a cycle or the graph being strongly connected.
\end{proof}

\coreinterleavedamount*
\begin{proof}
    Note that each interleaved walk with a core of length 1 is safe only if its core is a cut arc.
    So there are at most $O(n)$ interleaved safe walks with a core of length 1.
    Further, by Cairo et al.~\cite{cairo2020macrotigs}, it holds that there are at most $O(n)$ walks without forbidden paths that are interleaved with a core of length at least 2.
    Assume for a contradiction that any such walk $W$ could contain more than one non-cut arc.
    The non-cut arcs cannot be outside of the core, since all non-core arcs are the only outgoing or the only incoming arcs of some node.
    If there are at least two non-cut arcs in the core, we can construct an arc-covering closed walk $W'$ that does not contain $W$.
    Start with any arc-covering closed walk and repeat it twice.
    In the first repetition, replace any occurrence of the first non-cut arc of $W$ with a walk that avoids the non-cut arc.
    And in the second repetition, replace any occurrence of the last non-cut arc of $W$ with a walk that avoids the non-cut arc.
    Such avoiding walks exist since the avoided arcs are not cut arcs.
    Further, by avoiding an arc of $W$, they do not have $W$ as subwalk.
    The resulting walk $W'$ is arc-covering and closed, but does not have $W$ as subwalk.
    Hence, $W$ has a forbidden path by \Cref{lem:arc-to-node}.
    Finally, each walk without forbidden path that is interleaved with a core of length at least 2 produces at most two interleaved safe walks.
    Since non-interleaved walks without forbidden path cannot be broken into interleaved ones, there are at most $O(n)$ interleaved maximal safe walks.
\end{proof}

\end{document}